\documentclass[a4paper,envcountsame,envcountsect]{llncs}

\usepackage{style}
\newcommand{\pset}{\mathbb{N}^{P}}

\newcommand{\petri}{\mathcal{N}}
\newcommand{\petrituple}{\mathcal{N} = (P,T,F,\init)}

\newcommand{\init}{m_{init}}
\newcommand{\final}{m_{final}}

\newcommand{\upa}[1]{{\uparrow}{#1}}
\newcommand{\downa}[1]{{\downarrow}{#1}}

\newcommand{\funpre}[1][]{pre#1_{\petri}}
\newcommand{\funcpre}[1][]{cpre#1_{\petri}}
\newcommand{\pre}[2][]{\funpre#1(#2)}
\newcommand{\cpre}[2][]{\funcpre#1(#2)}
\newcommand{\prestar}[1]{\pre[^*]{#1}}
\newcommand{\prestarfinal}{\prestar{\upa{\final}}}
\newcommand{\Cov}{Cov_{\petri}}

\newcommand{\transition}[1]{\xrightarrow{#1}}
\newcommand{\transitionstar}{\transition{*}}
\newcommand{\transitionwithoutarg}{\rightarrow}

\newcommand{\Uk}{(U_k)_k}

\newcommand{\mytt}[1]{{\tt #1}}

\newcommand{\tool}[1]{\mytt{#1}}
\newcommand{\qcover}{\tool{QCover}}
\newcommand{\icover}{\tool{ICover}}
\newcommand{\preqcover}{\tool{QCover/Pp}}

\DeclareMathOperator{\Min}{Min}

\newcommand{\setN}{\mathbb{N}}
\newcommand{\setZ}{\mathbb{Z}}
\newcommand{\setQ}{\mathbb{Q}}
\newcommand{\moins}{\backslash}

\title{Occam's Razor Applied to the Petri Net Coverability Problem}
\author{%
  Thomas Geffroy
  \and
  Jérôme Leroux
  \and
  Grégoire Sutre
}

\institute{%
  Univ. Bordeaux \& CNRS, LaBRI, UMR 5800, Talence, France
}

\begin{document}

\maketitle

\pagestyle{plain}
\thispagestyle{plain}

\begin{abstract}
  The verification of safety properties for concurrent systems
often reduces to the \emph{coverability} problem for Petri nets.
This problem was shown to be \textsc{ExpSpace}-complete forty years ago.
Driven by the concurrency revolution,
it has regained a lot of interest over the last decade.
In this paper,
we propose a generic and simple approach to solve this problem.
Our method is inspired from the recent approach of
Blondin, Finkel, Haase and Haddad~\cite{Blondinetal16}.
Basically,
we combine forward invariant generation techniques for Petri nets
with backward reachability for well-structured transition systems.
An experimental evaluation demonstrates the efficiency of our approach.

\end{abstract}

\section{Introduction}

\paragraph{Context.}

The analysis of concurrent systems with unboundedly many processes
classically uses the so-called \emph{counter abstraction}~\cite{GS92}.
The main idea is to forget about the identity of each process,
so as to make processes indistinguishable.
Assuming that each process is modeled by a finite-state automaton,
it is then enough to count,
for each state $q$,
how many processes are in state $q$.
The resulting model is a Petri net,
with no a priori bound on the number of tokens.
The verification of a safety property on the original concurrent system
(e.g., mutual exclusion)
translates into a \emph{coverability} question on the Petri net:
Is it possible to reach a marking that is component-wise larger than a given marking?

\paragraph{Related work.}

Karp and Miller~\cite{KarpM69} proved in 1969 that
coverability is decidable (but their algorithm is not primitive recursive),
Lipton showed that it requires at least exponential space~\cite{Lipton76},
and Rackoff showed that it only requires exponential space~\cite{Rackoff78}.
Despite these somewhat negative results,
and driven by the concurrency revolution,
the coverability problem has regained a lot of interest over the last decade.
Recent efficient approaches include target set widening~\cite{KaiserKW14}
and structural analysis mixed with SMT solving~\cite{EsparzaLMMN14,Blondinetal16}.
We believe that the time is ripe to experiment with new ideas and prototypes
for coverability,
and to apply them to real-world concurrent systems.

\smallskip

Our work builds notably on~\cite{Blondinetal16},
which proposes a new approach to the coverability problem and its
implementation.
The approach of~\cite{Blondinetal16} is conceptually simple and exploits recent
advances in the theory of Petri nets as well as the power of modern
SMT-solvers.
In a nutshell, they leverage recent results on coverability in
continuous Petri nets~\cite{FracaH15} to over-approximate coverability under the
standard semantics: any configuration that is not coverable in a
continuous Petri net is also not coverable under the standard
semantics. This observation is then exploited inside a
backward-coverability framework~\cite{ACJT00}.

\paragraph{Our contribution.}

We present a generic backward coverability algorithm that
relies on downward-closed (forward) invariants to prune the
exploration of the state space.
Our algorithm is in fact a family of algorithms parametrized
by downward-closed invariants.
It generalizes the algorithm presented in~\cite{Blondinetal16} and
implemented in the promising tool $\qcover$.
We implemented our algorithm as a variant of $\qcover$
that we call $\icover$.
Whereas $\qcover$ is based on invariants obtained from recent results
on continuous Petri nets~\cite{FracaH15}, our
tool $\icover$ is based on two classical methods:
the \emph{state equation} for Petri nets,
and data-flow \emph{sign analysis}~\cite{Cousot:1977:POPL}.
On the 143 Petri net coverability questions that
$\qcover$ solved,
the tool $\qcover$ took 10318 seconds,
while $\icover$ used only 5517 seconds.

\paragraph{Outline.}

\cref{sec:petri-coverability} recalls the Petri net coverability
problem.
\cref{sec:backward-pruning,sec:icover-algorithm} present
our backward coverability algorithm with pruning based on
downward-closed invariants.
In \cref{sec:stateeq,sec:mark},
we recall two classical methods for computing invariants,
namely the state (in-)equation and sign analysis.
\cref{sec:bench} is dedicated to the experimental evaluation
of the tool \icover.
In \cref{sec:comparison},
we provide mathematical foundations for explaining our empirical good
results based on the notion of limit-reachability in continuous Petri
nets~\cite{DBLP:conf/apn/RecaldeTS99}.

\section{The Coverability Problem for Petri nets}
\label{sec:petri-coverability}

A Petri net is a tuple $\petrituple$ comprising a finite set of \emph{places} $P$, a finite set of \emph{transitions} $T$ disjoint of $P$, a \emph{flow} function $F$ from $(P \times T) \cup (T \times P)$ to~$\setN$, and an \emph{initial} marking $\init \in \pset$.
It is understood that $\pset$ denotes the set of total maps from $P$ to $\setN$.
Elements of $\pset$ are called \emph{markings}.
Intuitively,
a marking specifies how many \emph{tokens} are in each place of the net.
Tokens are consumed and produced through the firing of transitions.
%
%
A transition $t \in T$ may fire only if it is enabled,
meaning that each place $p$ contains at least $F(p, t)$ tokens.
Firing an enabled transition $t$ modifies the contents of each place $p$ by
first removing $F(p, t)$ tokens and then adding $F(t, p)$ tokens.
To clarify this intuitive description of the Petri net semantics,
we introduce, for each transition $t \in T$,
the $t$-step binary relation $\transition{t}$ over $\pset$,
defined by
$$
m \transition{t} m'
\ \Leftrightarrow \ 
\forall p \in P : m(p) \geq F(p, t) \wedge m'(p) = m(p) - F(p, t) + F(t, p)
$$
The one-step binary relation $\transitionwithoutarg$ is the union of these $t$-step relations.
Formally,
$m \transitionwithoutarg m' \Leftrightarrow \exists t \in T : m \transition{t} m'$.
The many-step binary relation $\transitionstar$ is the reflexive-transitive closure of $\transitionwithoutarg$.

\begin{example}
  \cref{figure:petri} depicts a simple Petri net $\petrituple$ with
  places $P = \{p_{1}, p_{2}, p_{3}\}$,
  transitions $T = \{t_{1}, t_{2}, t_{3}\}$ and
  flow function $F$ such that
  $F(p_{1},t_{1}) = 1$,
  $F(p_{2},t_{2}) = 1$,
  $F(p_{3},t_{3}) = 1$,
  $F(t_{1},p_{2}) = 1$,
  $F(t_{2},p_{3}) = 2$,
  $F(t_{3},p_{2}) = 2$, and
  $F(p, t) = F(t, p) = 0$ for all other cases.
  The initial marking is $\init = (1,0,0)$.
  The sequence of transitions $t_1 t_2 t_3$ may fire from the initial marking.
  Indeed,
  $(1,0,0) \transition{t_{1}} (0,1,0) \transition{t_{2}} (0,0,2) \transition{t_{3}} (0,2,1)$.
\end{example}


\begin{figure}[t]
  \centering
\begin{tikzpicture}[node distance=1.5cm,auto]

  \tikzstyle{place}=[circle,thick,draw=blue!75,fill=blue!20,minimum size=6mm]
  \tikzstyle{red place}=[place,draw=black!75,fill=black!20]
  \tikzstyle{transition}=[rectangle,thick,draw=black!75,
  			  fill=black!20,minimum size=4mm]

  \tikzstyle{every label}=[red]

  \begin{scope}
    \node [place,tokens=1,label=above:$p_{1}$]  (p1)                                    {};
    \node [transition]                          (t1)  [right of=p1, xshift=5mm]         {$t_{1}$};
    \node [place,label=above left:$p_{2}$]      (p2)  [right of=t1, xshift=5mm]         {};
    \node [transition]                          (t2)  [above right of=p2, xshift=8mm]   {$t_{2}$};
    \node [transition]                          (t3)  [below right of=p2, xshift=8mm]   {$t_{3}$};
    \node [place,label=above right:$p_{3}$]     (p3)  [below right of=t2, xshift=8mm]   {};

    \path (t1)
    edge [pre, left]    (p1)
    edge [post, right]  (p2);

    \path (t2)
    edge [pre,  bend right]          (p2)
    edge [post, bend left]  node {2} (p3);

    \path (t3)
    edge [pre,bend right]           (p3)
    edge [post, bend left] node {2} (p2);
  \end{scope}

\end{tikzpicture}
\caption{Simple Petri net example}
  \label{figure:petri}
\end{figure}
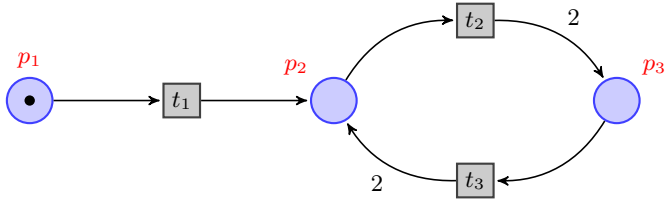

One of the most fundamental verification questions on Petri nets is coverability.
In its simplest form,
the coverability problem asks whether it is possible,
by firing a sequence of transitions,
to put a token in a given place.
In essence,
the coverability problem for Petri nets corresponds to the
control-state reachability problem for other models of computation,
such as counter machines,
which are equipped with control states.
The formal definition of coverability relies on a partial order over markings,
defined hereafter.

\smallskip

Let $\leq$ denote the usual total order on $\setN$.
We extend $\leq$ over $\pset$ component-wise,
by $m \leq m' \Leftrightarrow \forall p \in P : m(p) \leq m'(p)$.
This extension is a partial order over $\pset$.
Given two markings $m$ and $m'$ in $\pset$,
we say that $m$ \emph{covers} $m'$ when $m \geq m'$.
The \emph{coverability problem} asks,
given a Petri net $\petrituple$ and a \emph{target} marking $\final \in \pset$,
whether there exists a marking $m \in \setN^{P}$ such that
$\init \transitionstar m$ and $m \geq \final$.
The main goal of this paper is to provide a simple,
yet efficient procedure for solving this problem.
Our method is inspired from the recent approach of~\cite{Blondinetal16}.
Basically,
we combine forward invariant generation techniques for Petri nets
with backward reachability for well-structured transition systems~\cite{ACJT00,DBLP:journals/tcs/FinkelS01}.
Before delving into the details,
we need some additional notations.


\smallskip

For a transition $t \in T$ and a set $S \subseteq \pset$ of markings,
we let $\pre[^t]{S}$ denote the predecessors of $S$ via the transition $t$.
Similarly,
$\pre{S}$ and $\pre[^*]{S}$ denote the one-step and many-step predecessors of $S$,
respectively.
Formally,
the functions $\funpre[^t]$, $\funpre$ and $\funpre[^*]$
from $2^{\pset}$ to $2^{\pset}$ are defined by
$$
\begin{array}{r@{\ \ }c@{\ \ }l}
  \pre[^t]{S}  & = & \{m \in \pset \mid \exists m' \in S : m \transition{t} m'\}\\
  \pre{S}     & = & \{m \in \pset \mid \exists m' \in S : m \transitionwithoutarg m'\}\\
  \pre[^*]{S} & = & \{m \in \pset \mid \exists m' \in S : m \transitionstar m'\}
\end{array}
$$

Given a subset $S \subseteq \pset$ of markings,
we let $\upa{S}$ and $\downa{S}$ denote
its \emph{upward closure} and \emph{downward closure}, respectively.
These are defined by
$$
\begin{array}{r@{\ \ }c@{\ \ }l}
\upa{S}   & = & \{u \in \pset \mid \exists m \in S : u \geq m\}\\
\downa{S} & = & \{d \in \pset \mid \exists m \in S : d \leq m\}
\end{array}
$$
A subset $S \subseteq \pset$ is called \emph{upward-closed} when $S = \upa{S}$,
and it is called \emph{downward-closed} when $S = \downa{S}$.

\begin{notation}
  For the remainder of the paper,
  to avoid clutter,
  we will simply write $m$ in place of $\{m\}$ for singletons,
  when this causes no confusion.
\end{notation}

Recall that the coverability problem asks whether
$\init \transitionstar m \geq \final$ for some marking $m \in \setN^{P}$.
This problem is equivalently phrased as the question whether
$\init$ belongs to $\pre[^*]{\upa{\final}}$.
This formulation can be seen as a backward analysis question.
We may also phrase the coverability problem in terms of a forward analysis
question, using the notion of coverability set.

\smallskip

Given a Petri net $\petrituple$,
the \emph{coverability set} of $\petri$ is the set
$\Cov = \downa{\{m \in \pset \mid \init \transitionstar m\}}$.
It is readily seen that the coverability problem is equivalent to
the question whether $\final$ belongs to $\Cov$.
We are now equipped with the necessary notions to present our
mixed forward/backward approach for the coverability problem.

\section{Backward Coverability Analysis with Pruning}
\label{sec:backward-pruning}
\label{sec:pruning}

We now present our method to solve the coverability problem for Petri nets.
This section gives the mathematical foundations of our approach,
with no regard for implementability.
We will focus on the implementation of this approach in \cref{sec:icover-algorithm}.

\smallskip

The classical backward reachability approach for the coverability
problem~\cite{ACJT00,DBLP:journals/tcs/FinkelS01}
consists in computing a growing sequence $U_0 \subseteq U_1 \subseteq \cdots$
of upward-closed subsets of $\pset$ that converges to
$\prestarfinal$.
Here,
we modify this growing sequence in order to leverage an a priori known
over-approximation of the coverability set.
In practice,
this means that we narrow the backward reachability search by
pruning some markings that are known to be not coverable.

\smallskip

An \emph{invariant} for a Petri net $\petrituple$ is any subset $I \subseteq \pset$
that contains every reachable marking, i.e.,
every marking $m$ with $\init \transitionstar m$.
Observe that a downward-closed subset of $\pset$ is an invariant of $\petri$
if, and only if,
it contains $\Cov$.
\cref{sec:stateeq,sec:mark} will discuss the automatic generation of
downward-closed invariants.

\smallskip

For the remainder of this section,
we consider a Petri net $\petrituple$ and
we assume that we are given a downward-closed invariant $I$ for $\petri$.
We introduce the sequence $U_0, U_1, \ldots$ of subsets of $\pset$
defined as follows:
$$
\begin{array}{r@{\ \ }c@{\ \ }l}
  U_{0}   & = & \upa{(\final \cap I)} \\
  U_{k+1} & = & \upa{(\pre{U_k} \cap I)} \cup U_k
\end{array}
$$
Observe that each $U_k$ is upward-closed and that
the sequence $\Uk$ is growing for inclusion.
On the contrary to the classical backward reachability
approach~\cite{ACJT00,DBLP:journals/tcs/FinkelS01},
$U_{k+1}$ does not consider all one-step predecessors of $U_k$,
but discards those that are not in $I$.
Note that by taking $I = \pset$,
which is trivially a downward-closed invariant,
we obtain the same growing sequence as in the classical backward reachability
approach~\cite{ACJT00,DBLP:journals/tcs/FinkelS01}.
The two following lemmas show that we can use the sequence $\Uk$
to solve the coverability problem.

\begin{restatable}{lemma}{lemukstationary}
  \label{lem:uk-stationary}
  The sequence $\Uk$ is ultimately stationary.
\end{restatable}

\vspace{-1em} %
\begin{restatable}{lemma}{lemukcov}
  \label{lem:uk-cov}
  It holds that
  $\final \in \Cov$ if, and only if, $\init \in \bigcup_k U_k$.
\end{restatable}

We have presented in this section a growing sequence of upward-closed subsets
of markings that is ultimately stationary and whose limit contains enough
information to solve the coverability problem.
Our next step is to transform this sequence into an algorithm.

\section{The $\mathtt{ICover}$ Algorithm}
\label{sec:icover-algorithm}

In this section,
we turn the growing sequence $\Uk$ of upward-closed subsets of markings
defined in \cref{sec:backward-pruning} into an algorithm.
Of course,
we cannot directly compute the sets $U_k$ since they may be infinite
(in fact, they are either empty or infinite).
Instead,
we will compute finite sets $B_k \subseteq \pset$ such that $U_k = \upa{B_k}$.
The existence of such finite sets is guaranteed by the following lemma.
A \emph{basis} of an upward-closed subset $U \subseteq \pset$ is
any set $B \subseteq \pset$ such that $U = \upa{B}$.
Recall that a \emph{minimal} element of a subset $S \subseteq \pset$ is
any $m \in S$ such that $u \leq m \Rightarrow u = m$ for every $u \in S$.

\begin{restatable}{lemma}{lemwqomin}
  \label{lem:wqo-min}
  For every subset $S \subseteq \pset$,
  the set $\Min S$ of its minimal elements is finite and satisfies
  $\upa{S} = \upa{\Min S}$.
\end{restatable}

\begin{corollary}
  Every upward-closed subset $U \subseteq \pset$ admits a finite basis.
\end{corollary}

We still need to show how to compute a finite basis of $U_{k+1}$ from a finite basis of $U_k$.
To this end we introduce, for each transition $t \in T$,
the \emph{covering predecessor} function $\funcpre[^t] : \pset \rightarrow \pset$ defined by
$$
\cpre[^t]{m}(p) \ = \ F(p, t) + \max (0, m(p) - F(t, p))
$$
Informally,
$\cpre[^t]{m}$ is the least marking that can cover $m$ in one step
by firing the transition $t$.
This property will be formally stated in \cref{lem:cpre}.
The function $\funcpre[^t]$ is extended to sets of markings by
$\cpre[^t]{S} = \{\cpre[^t]{m} \mid m \in S\}$.

\begin{restatable}{lemma}{lemcpre}
  \label{lem:cpre}
  It holds that $\pre[^t]{\upa{m}} = \upa{\cpre[^t]{m}}$
  for every marking $m \in \pset$.
\end{restatable}

The previous lemma can easily be extended to sets of markings.
We extend it further,
in \cref{lem:cpre-inv},
to bridge the gap with the definition of $\Uk$.
The lemma shows how to compute a finite basis of $U_{k+1}$
from a finite basis of $U_k$.

\begin{lemma}
  \label{lem:cpre-inv}
  Let $I$ be a downward-closed invariant for $\petri$.
  For every subset $S \subseteq \pset$,
  it holds that $\upa{\pre[^t]{(\upa{S}) \cap I}} = \upa{(\cpre[^t]{S} \cap I)}$.
\end{lemma}
\begin{proof}
  The straightforward extension of \cref{lem:cpre} to sets of markings
  shows that
  $\pre[^t]{\upa{S}} = \upa{\cpre[^t]{S}}$ for every subset $S \subseteq \pset$.
  Moreover,
  it is readily seen that,
  for every subset $S \subseteq \pset$,
  $\upa{((\upa{S}) \cap I)} = \upa{(S \cap I)}$.
  This property follows from the assumption that $I$ is downward-closed.
  We derive that
  $$
  \begin{array}{r@{\ \ }c@{\ \ }l}
    \upa{(\cpre[^t]{S} \cap I)}
    & = & \upa{((\upa{\cpre[^t]{S}}) \cap I)}\\
    & = & \upa{(\pre[^t]{\upa{S}} \cap I)}
  \end{array}
  $$
  This concludes the proof of the lemma.
  \qed
\end{proof}

\begin{algorithm}[t]
  \NoCaptionOfAlgo
  \DontPrintSemicolon
  \Input{%
    A Petri Net $\petrituple$,
    a target marking $\final \in \pset$ and
    a downward-closed invariant $I$ for $\petri$.
  }
  \Output{%
    Whether there exists a marking $m \in \pset$ such that
    $\init \transitionstar m$ and $m \geq \final$.
  }

  \Begin{
    \eIf{$\final \in I$}{
      $B \leftarrow \{\final\}$\;
    }{
      $B \leftarrow \emptyset$\;
    }
    \While{$\init \not\in \upa{B}$}{
      \label{line:while-start}
      $N \leftarrow \{\cpre[^t]{m} \mid t \in T, m \in B\} \setminus \upa{B}$
      \tcc*{new predecessors}
      \label{line:assign:N}
      $P \leftarrow N \cap I$
      \label{line:assign:P}
      \tcc*{prune uncoverable markings}
      \If{$P = \emptyset$}{
        \label{line:snapshot}
        \Return{$\mathtt{False}$}\;
        \label{line:return:false}
      }
      $B \leftarrow \Min (B \cup P)$\;
      \label{line:while-end}
    }
    \Return{$\mathtt{True}$}\;
    \label{line:return:true}
  }
  \caption{$\mathtt{ICover}(\petri, \final, I)$}
  \label{algo:icover}
\end{algorithm}

The previous lemma leads to a backward coverability algorithm,
called $\mathtt{ICover}$ and presented on page~\pageref{algo:icover}.
Basically,
this procedure symbolically computes the growing sequence $\Uk$
of upward-closed sets.
Let us make the relationship between the procedure and the sequence $\Uk$
more precise.
Consider an input instance $(\petri, \final, I)$ of $\mathtt{ICover}$.
Since the procedure is deterministic,
$\mathtt{ICover}(\petri, \final, I)$ has a unique maximal execution,
that
either terminates
(at line~\ref{line:return:false} or~\ref{line:return:true})
or iterates the $\mathbf{while}$ loop
(lines~\ref{line:while-start}--\ref{line:while-end}) indefinitely.
Let $\ell_B, \ell_P \in \setN \cup \{\infty\}$ denote the numbers of executions
of lines~\ref{line:while-start} and~\ref{line:snapshot}, respectively.
It is understood that $l_P \leq l_B \leq l_P + 1$,
with the convention that $\infty + 1 = \infty$.
Let $(B_k)_{k < \ell_B}$ and $(P_k)_{k < \ell_P}$ denote the successive values
at lines~\ref{line:while-start} and~\ref{line:snapshot}
of the variables $B$ and $P$, respectively.

\begin{restatable}{lemma}{lemukbk}
  \label{lem:uk-bk}
  For every $k$ with $0 \leq k < \ell_B$,
  the set $B_k$ is a finite basis of $U_k$.
  For every $k$ with $0 \leq k < \ell_P$,
  the set $P_k$ is a finite basis of $\upa{(U_{k+1} \setminus U_k)}$.
\end{restatable}

\vspace{-1em} %
\begin{restatable}{theorem}{thmicover}
  The procedure $\mathtt{ICover}$ terminates on every input and is correct.
\end{restatable}

\begin{remark}\label{rem:prepro}
  Petri nets obtained by translation from high-level concurrent programs
  often contain transitions that cannot be fired from any
  reachable marking.
  Downward-closed invariants can be used in a
  pre-processing algorithm to filter out some of them.
  Basically, if a transition $t$ is not
  enabled in any marking of an invariant $I$, it
  can be safely removed without modifying the coverability set.
  Algorithmically,
  when $I$ is downward-closed,
  detecting such a property just reduces to a membership problem in $I$.
  In fact a transition $t$ is enabled in a
  downward-closed set of markings $D$ if, and only if, $D$ contains the marking $m_t$ defined by
  $m_t(p)=F(p,t)$ for every place $p$.
\end{remark}

%
%
%
%
%

The algorithm $\mathtt{ICover}$ is parametrized by an a priori known
downward-closed invariant that is given as input.
On the one hand,
this invariant needs to be precise enough to discard markings (at line~\ref{line:assign:P}) and accelerate the main loop.
On the other hand,
we need to decide efficiently whether a marking is in the invariant,
to avoid slowing down the main loop.
The next two sections show how to generate downward-closed invariants
with efficient membership testing.

\section{State Inequation for Downward-Closed Invariants}\label{sec:stateeq}
The state equation provides a simple over-approximation of  Petri net reachability relations that was successfully used in two recent algorithms for deciding the coverability problems \cite{EsparzaLMMN14,Blondinetal16}. This equation is obtained by introducing the total function $\Delta(t)$ in $\setZ^P$ called the \emph{displacement} of a transition $t$ and defined for every place $p$ by $\Delta(t)(p)=F(t,p)-F(p,t)$. Let us assume that a marking $\final$ is in the coverability set of a Petri net $\petri$. It follows that there exists a word $t_1\ldots t_k$ of transitions and a marking $m\geq \final$ such that $\init\xrightarrow{t_1}\cdots\xrightarrow{t_k}m$. We derive the following relation:
$$\init+\Delta(t_1)+\cdots+\Delta(t_k)=m\geq \final$$
By reordering the sum $\Delta(t_1)+\cdots+\Delta(t_k)$, we can group together the displacements $\Delta(t)$ corresponding to the same transition $t$. Denoting by $\lambda(t)$ the number of occurrences of $t$ in the word $t_1\ldots t_k$, we get:
\begin{equation}\label{eq:stateeq}
  \init+\sum_{t\in T}\lambda(t)\Delta(t)\geq \final
\end{equation}
The relation~\eqref{eq:stateeq} is called the \emph{state inequation} for the coverability problem. Notice that a similar equation can be derived for the reachability problem by replacing the inequality by an equality. We do not consider this equality in the sequel since we restrict our attention to the coverability problem. We introduce the following set $I_S$ where $\setQ_{\geq 0}$ is the set of non-negative rational numbers.
\begin{equation}\label{eq:MSdef}
  I_S=\{m\in\pset \mid \exists \lambda\in \setQ_{\geq 0}^T: \init+\sum_{t\in T}\lambda(t)\Delta(t)\geq m\}
\end{equation}

\begin{proposition}
  The set $I_S$ is a downward-closed invariant with a polynomial-time membership problem.
\end{proposition}

A more precise downward-closed invariant can be obtained by requiring that $\lambda\in\setN^T$. In particular, the pruned backward algorithm presented in \cref{sec:icover-algorithm} should produce smaller sets of configurations with this more precise invariant. In practice, we do not observe any significant improvement on a large set of benchmarks. Moreover, whereas the membership problem of a marking $m$ is decidable in polynomial time when $\lambda$ ranges over $\setQ_{\geq 0}^T$, the problem becomes NP-complete when $\lambda$ is restricted to $\setN^T$.

\newcommand{\prop}[2]{\operatorname{prop}_{#1}(#2)}
\section{Sign Analysis for Downward-Closed Invariants}\label{sec:mark}\label{sec:sign-analysis}
In this section we introduce a downward-closed invariant based on data-flow sign analysis~\cite{Cousot:1977:POPL}. Rephrased in the context of Petri nets, an invariant $I$ is said to be \emph{inductive} if $m\xrightarrow{t}m'$ and $m\in I$ implies $m'\in I$. Sign analysis then reduces to the computation of the maximal (for the inclusion) set $Z$ of places such that the following set $I_Z$ is an inductive invariant:
\begin{equation}\label{eq:MZdef}
  I_Z=\{m\in\setN^P \mid \bigwedge_{p\in Z}m(p)=0\}
\end{equation}
The unicity of that set is immediate since the class of sets $Z$ such that $I_Z$ is an invariant is clearly closed under union. In the sequel, $Z$ denotes the maximal set satisfying this property, and this maximal set is shown to be computable in polynomial time thanks to a fixpoint propagation. We introduce the operator $\operatorname{prop}_{t}:2^P\rightarrow2^P$ associated to a transition $t$ and defined for any set $Q$ of places as follows:
$$\prop{t}{Q}=\begin{cases}
  \{q\in P \mid F(t,q)>0\} & \text { if  }\bigwedge_{p\in P\moins Q}F(p,t)=0\\
  \emptyset & \text{ otherwise}
\end{cases}$$
Intuitively, if $t$ is a transition such that $\bigwedge_{p\in P\moins Q}F(p,t)=0$ then from a marking with large number of tokens in each place of $Q$, it is possible to fire $t$. In particular places $q$ satisfying $F(t,q)>0$ cannot be in $Z$. This property is formally stated by the following lemma.
\begin{restatable}{lemma}{lemprop}
  \label{lem:prop}
  We have $\prop{t}{Q}\subseteq P\moins Z$ for every set $Q\subseteq P\moins Z$.
\end{restatable}

The set $Z$ can be computed as a fixpoint by introducing the non-decreasing sequence $Q_0,Q_1,\ldots$ of places defined as follows:
\begin{align*}
  Q_0&=\{q\in P \mid \init(q)>0\}\\
  Q_{k+1}&=Q_k\cup\bigcup_{t\in T}\prop{t}{Q_k}
\end{align*}
Let us notice that the set $Q=\bigcup_{k\geq 0}Q_k$ is computable in polynomial time. The following lemma shows that $Q$ provides the set $Z$ as a complement.
\begin{restatable}{lemma}{lemzpmoinsq}
  We have $Z=P\moins Q$.
\end{restatable}

\begin{corollary}
  The set $Z$ is computable in polynomial time.
\end{corollary}

\section{Experimental Evaluation}\label{sec:bench}

We implemented our approach using the $\qcover$~\cite{Blondinetal16} tool as a starting point.
This tool,
which implements a backward coverability algorithm for Petri nets,
is written in Python and relies on the SMT-solver $z3$~\cite{DBLP:conf/tacas/MouraB08}.
$\qcover$ also uses some other heuristics that we kept unchanged.
$\qcover$ was competitive with others tools especially for uncoverable Petri net.
Only the $\tool{BFC}$ tool performs significantly better on coverable Petri net.
We have made two modifications to $\qcover$.
First, we have added a pre-processing step (see \cref{rem:prepro}) based on sign analysis.
Second,
we have replaced their pruning technique,
which is based on coverability in continuous Petri nets,
by the one of our algorithm $\icover$ presented in \cref{sec:icover-algorithm}.
$\icover$ is available as a patch~\cite{icover} for $\qcover$~\cite{qcover}.

\smallskip

To test our implementation,
we used the same benchmark as $\tool{Petrinizer}$~\cite{EsparzaLMMN14} and $\qcover$~\cite{Blondinetal16}.
It comprises models from various sources: \tool{Mist}~\cite{mist},
$\tool{BFC}$~\cite{KaiserKW14},
${\tt Erlang}$ programs abstracted into Petri nets~\cite{DOsualdoKO13},
as well as so-called ${\tt medical}$ and ${\tt bug\_tracking}$ examples~\cite{EsparzaLMMN14}.
We let each tool work for 2000 seconds in a machine on Ubuntu Linux 14.04 with Intel(R) Core(TM) i7-4770 CPU at 3.40GH with 16 GB of memory for each benchmark.
The computation times are the sum of the $\mytt{system}$ and $\mytt{user}$ times.
Overall $\qcover$ solved 106 uncoverable instances on 115 Petri net and 37 coverable problems on 61 Petri nets.
$\icover$ was able to find one more coverable instance.
In fact calling $\qcover$ on the Petri net computed by the pre-processing,
that we will call $\preqcover$, can even solve one more uncoverable instance than $\icover$.
On the 143 instances that $\qcover$ solved, the tool took 10318 seconds,
$\preqcover$ used 6479 seconds,
and $\icover$ used only 5162 seconds.

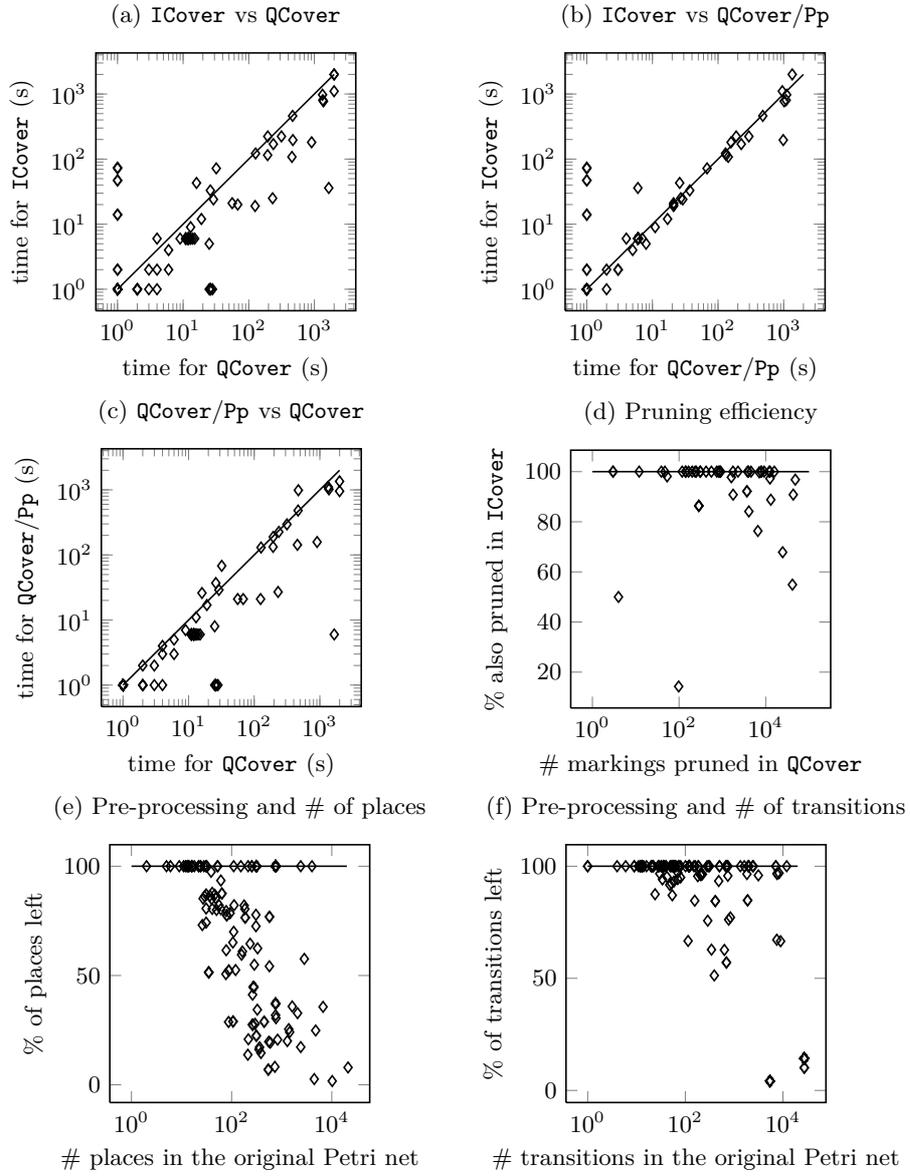
\begin{figure}[h!]
  \begin{subfigure}{.5\textwidth}
    \begin{tikzpicture}
      \begin{loglogaxis}[%
        title=(a) $\icover$ vs $\qcover$,
        width=5cm,
        height=5cm,
        grid=none,
        xlabel={time for $\qcover$ (s)},
        ylabel={time for $\icover$ (s)}
        ]
        \addplot[only marks,mark=diamond,scatter/use mapped color= {draw=black,fill=black}] file {plots/friday2prelimit_to_noqcover.dat};

        \addplot [mark=none,black,domain=1:2000] {x};
      \end{loglogaxis}
    \end{tikzpicture}    
    \label{fig:icoverqcover}
  \end{subfigure}  
  \hfill
  \begin{subfigure}{.5\textwidth}
    \begin{tikzpicture}
      \begin{loglogaxis}[title=(b) $\icover$ vs $\preqcover$,
        width=5cm,
        height=5cm,
        xlabel={time for $\preqcover$ (s)},
        ylabel={time for $\icover$ (s)},  ]
        \addplot[only marks,mark=diamond,scatter/use mapped color= {draw=black,fill=black}] file {plots/friday2prelimit_to_preqcover.dat};
        \addplot [mark=none,black,domain=1:2000] {x};
      \end{loglogaxis}
    \end{tikzpicture}
    \label{fig:icoverpreqcover}
  \end{subfigure}  
  
  \begin{subfigure}{.5\textwidth}
  \end{subfigure}  
  \begin{subfigure}{.5\textwidth}
    \begin{tikzpicture}
      \begin{loglogaxis}[title=(c) $\preqcover$ vs $\qcover$,
        width=5cm,
        height=5cm,
        xlabel={time for $\qcover$ (s)},
        ylabel={time for $\preqcover$ (s)},  ]
        \addplot[only marks,mark=diamond,scatter/use mapped color= {draw=black,fill=black}] file {plots/friday2preqcover_to_noqcover.dat};

        \addplot [mark=none,black,domain=1:2000] {x};
      \end{loglogaxis}
    \end{tikzpicture}
    \label{fig:preqcoverqcover}
  \end{subfigure}  
  \hfill
  \begin{subfigure}{.5\textwidth}
    \begin{tikzpicture}
      \begin{semilogxaxis}[title=(d) Pruning efficiency,
        width=5cm,
        height=5cm,
        xlabel={\# markings pruned in $\qcover$},
        ylabel={\% also pruned in $\icover$},  ]
        \addplot[only marks,mark=diamond,scatter/use mapped color= {draw=black,fill=black}] file {plots/glitch.dat};        
        \addplot [mark=none,black,domain=1:100000] {100.0};
      \end{semilogxaxis}
    \end{tikzpicture}
    \label{fig:comparison}
  \end{subfigure}  
  \begin{subfigure}{.5\textwidth}
    \centering 
    \begin{tikzpicture}
      \begin{semilogxaxis}[title=(e) Pre-processing and \# of places, %
        width=5cm,
        height=5cm,
        xlabel={\# places in the original Petri net},
        ylabel={\% of places left},  ]
        \addplot[only marks,mark=diamond,scatter/use mapped color= {draw=black,fill=black}] file {plots/dimensionspourcent.dat};
        
        \addplot [mark=none,black,domain=1:20000] {100};
      \end{semilogxaxis}
    \end{tikzpicture}
    \label{fig:predimension}
  \end{subfigure}  
  \begin{subfigure}{.5\textwidth}
  \centering
    \begin{tikzpicture}
      \begin{semilogxaxis}[title=(f) Pre-processing and \# of transitions, %
        width=5cm,
        height=5cm,
        xlabel={\# transitions in the original Petri net},
        ylabel={\% of transitions left},  ]
        \addplot[only marks,mark=diamond,scatter/use mapped color= {draw=black,fill=black}] file {plots/transitionspourcent.dat};        
        \addplot [mark=none,black,domain=1:20000] {100.0};
      \end{semilogxaxis}
    \end{tikzpicture}    
    \label{fig:pretransitions}
\end{subfigure}
\caption{Experimental results for $\icover$, $\qcover$ and $\preqcover$}\label{fig:exp}
\end{figure}

\smallskip

\cref{fig:exp}(a) shows the comparison between $\icover$ and $\qcover$ in time. The straight line represents when the two tools took the same time.
Each dot represents a coverability question.
When the dot is under the line, it means that $\icover$ was faster than $\qcover$ and conversely.
There are three instances where $\qcover$ performs very well, under a second, and where $\icover$ took a few tens of seconds to answer. For the three cases, the formula used by $\qcover$ for coverability in $\mathbb{Q}$ was enough to discard the target as uncoverable and it didn't have to enter in the while loop. But $\icover$ wasn't able to discard the target and had to enter the while loop in the three cases.
We also see two dots above the line at the middle of the figure. The pre-processing took respectively 12 and 45 seconds while the initial Petri net was solved by $\qcover$ in respectively 16 and 33 seconds. The pre-processing has not been optimized yet, and it could probably run faster.

\smallskip

\cref{fig:exp}(b) and (c) show the intermediate comparisons: $\icover$ versus $\preqcover$ and $\preqcover$ versus $\qcover$. We can observe that the pre-processing has a major impact on the good performance of $\icover$ compared to $\qcover$.

\smallskip

\cref{fig:exp}(e) and \cref{fig:exp}(f) aims to show the effect of the pre-processing on the size of Petri nets. The former show the percentage of places left after pre-processing. Some Petri nets kept all their places but others were left with only 2.5\% of their initial places. And most of Petri nets lost a significant number of places.
The latter shows the percentages of transitions left after the pre-processing. Overall less transitions were cut than places.
Half of the Petri nets kept all their transitions, but some were left with only 4\% of their initial transitions.

\smallskip

\cref{fig:exp}(d) compares the efficiency of pruning between $\icover$ and $\qcover$.
Again, each dot represents a coverability question.
As discussed in \cref{sec:comparison},
$\qcover$ always prunes at least as many markings as $\icover$
(but at the expense of more complex pruning tests).
A value of 100\% means that $\icover$ was able to prune the same markings as $\qcover$.
It turns out that on most instances,
this perfect value of 100\% is obtained.
This is rather surprising at first sight,
and warrants an investigation,
which is the focus of the next section.

\section{Comparison with Continuous Petri Net}\label{sec:comparison}
Continuous Petri nets are defined like Petri nets except that transitions can be fired a non-negative rational number of times. The firing of such a transition produces markings with non-negative rational numbers of tokens. Under such a semantics, called the \emph{continuous semantics}, the reachability problem was recently proved to be decidable in polynomial time~\cite{FracaH15}. Based on this observation, the tool $\qcover$ implements the pruning backward coverability algorithm presented in \cref{sec:pruning} with a downward-closed invariant derived from the continuous semantics.
Whereas this invariant is more precise than the downward-closed invariant obtained from the state inequation introduced in \cref{sec:stateeq}, we have seen in \cref{sec:bench} that such an improvement is overall not useful in practice for the pruning backward algorithm. In this section, we provide a simple structural condition on Petri nets in such a way the two kinds of downward-closed invariants derived respectively from the continuous semantics and the state inequation are ``almost'' equal. This structural condition is shown to be natural since it is fulfilled by the Petri nets obtained after the pre-processing introduced in \cref{rem:prepro}.

\medskip

A \emph{continuous marking} is a mapping $m\in\setQ_{\geq0}^P$ where $\setQ_{\geq 0}$ denotes the set of non-negative rational numbers, and $P$ the set of places. Given $r\in\setQ_{\geq 0}$ and a transition $t$, the continuous $rt$-step binary relation $\dotarrow{rt}$ over the continuous markings is defined by
$$
m \dotarrow{rt} m'
\ \Leftrightarrow \ 
\forall p \in P : m(p) \geq r.F(p, t) \wedge m'(p) = m(p) - r.F(p,t) + r.F(t,p)
$$
The one-step continuous binary relation $\dotarrow{}$ is the union of these $rt$-step relations. Formally, $m\dotarrow{}m'$ if there exists $r\in\setQ_{\geq 0}$ and $t\in T$ such that $m\dotarrow{rt}m'$. The many-step continuous binary relation $\dotarrow{*}$ is the reflexive-transitive closure of $\dotarrow{}$. We also introduce the binary relation $\dotarrow{\infty}$ defined over the continuous markings by $m \dotarrow{\infty}m'$ if there exists a sequence $(m_k)_{k\geq 0}$ of continuous markings that \emph{converges} towards $m'$ with the classical topology on $\setQ_{\geq 0}^P$ and such that $m\dotarrow{*}m_k$ for every $k$. 

\begin{example}\label{ex:CS}
  Let us look back at the simple Petri net $\petri$ depicted in \cref{figure:petri}. For every positive natural number $k$, we have: $$(1,0,0)\dotarrow{\frac{1}{k}t_1}(1-\frac{1}{k},\frac{1}{k},0)\dotarrow{\frac{1}{k}t_2\frac{1}{k}t_3}(1-\frac{1}{k},\frac{2}{k},\frac{1}{k})\cdots \dotarrow{\frac{1}{k}t_2\frac{1}{k}t_3}(1-\frac{1}{k},1+\frac{1}{k},1)$$
  It follows that $(1,0,0)\dotarrow{\infty}(1,1,1)$. Notice that the relation $(1,0,0)\dotarrow{*}(1,1,1)$ does not hold.
\end{example}

The downward-closed invariant used in the tool $\qcover$ for implementing the pruning backward algorithm is defined as follows:
\begin{equation}\label{eq:MCdef}
  I_C=\{m\in\setN^P \mid \exists m'\in\setQ_{\geq 0}^P:\init\dotarrow{*}m'\geq m\}
\end{equation}

Recall that in \cref{sec:stateeq} we introduced the set $I_S$ for denoting the downward-closed invariant derived from the state inequation. The following result\footnote{The statement of Theorem~7 in \cite{DBLP:conf/apn/RecaldeTS99} is wrong since it is based on a too strong definition of limit-reachability. However, the proof becomes correct with our definitions and notations.}
provides a characterization of that invariant when the Petri net satisfies a structural condition.
\begin{theorem}[{\cite[Theorem~7]{DBLP:conf/apn/RecaldeTS99}}]\label{thm:limit}
  If every transition is fireable from the downward-closed invariant $I_Z$ introduced in \cref{sec:mark}, we have:
  \begin{equation}\label{eq:MSbis}
    I_S=\{m\in\setN^P \mid \exists m'\in\setQ_{\geq 0}^P:\init\dotarrow{\infty}m'\geq m\}
  \end{equation}
\end{theorem}

The two equalities \cref{eq:MCdef} and \cref{eq:MSbis} show that $I_S$ and $I_C$ are very similar for Petri nets satisfying the structural condition stated in \cref{thm:limit}. This condition will be fulfilled by the Petri nets produced by the pre-processing algorithm introduced in \cref{rem:prepro}. Notice that even if the membership problem in $I_S$ and $I_C$ are both decidable in polynomial time, the extra computational cost for deciding the membership problem for the invariant $I_C$, even for efficient SMT solvers like \texttt{Z3}, is not neglectable. Naturally, if a marking is in $I_C$ then it is also in $I_S$, and the converse property is false in general as shown by \cref{ex:CS}. However, in practice, we observed that configurations that are in $I_S$ are very often also in $I_C$ (see \cref{fig:exp}(d)), as already mentioned in \cref{sec:bench}.

\section{Conclusion}

Petri nets have recently been used as low-level models for model-checking
concurrent systems written in high-level programming
languages~\cite{dkkw2011-cav,DOsualdoKO13}.
The original verification question on the concurrent program reduces to
a coverability question on the resulting Petri net.
We have proposed in this paper a family of simple coverability algorithms
parametrized by downward-closed invariants.
As future work,
we intend to look for classes of downward-closed invariants
with a good tradeoff between precision and efficient membership.

\bibliographystyle{splncs03}
\bibliography{reference}

\clearpage
\appendix

\section{Proofs for \cref{sec:backward-pruning}}
\label{appendix:sec:backward-pruning}
\lemukstationary*
\begin{proof}
  The partial order $\leq$ on $\pset$ is a well-quasi-order by Dickson's Lemma.
  Therefore,
  every growing sequence of upward-closed sets is ultimately stationary
  (see, e.g., \cite[Lemma~2.4]{DBLP:journals/tcs/FinkelS01}).
  The lemma follows from the observation that
  each $U_k$ is upward-closed and that
  the sequence $\Uk$ is growing.
  \qed
\end{proof}

\lemukcov*
\begin{proof}
  If $\final \in \Cov$ then $\init \transitionstar m \geq \final$
  for some marking $m$ in $\pset$.
  Since $\init \transitionstar m$,
  there exists $m_0, \ldots, m_n \in \pset$
  such that $\init = m_n$,
  $m_n \transitionwithoutarg m_{n-1} \cdots \transitionwithoutarg m_0$
  and $m_0 \geq \final$.
  First observe that $m_i \in I$ for every $i \in \{0, \ldots, m\}$
  because $I$ is an invariant for $\petri$.
  Moreover, $\final \in I$ since
  $I$ is downward-closed, $m_0 \geq \final$ and $m_0 \in I$.
  We prove, by induction on $i$, that $m_i \in U_i$ for all $i \in \{0, \ldots, n\}$.
  The basis $m_0 \in U_0$ follows from the facts that
  $m_0 \geq \final$ and $\final \in I$.
  For the induction step,
  let $i \in \{0, \ldots, n-1\}$ and assume that $m_i \in U_i$.
  Recall that $m_{i+1} \in I$ and $m_{i+1} \transitionwithoutarg m_i$.
  It follows that $m_{i+1} \in (\pre{U_i} \cap I) \subseteq U_{i+1}$.
  We have thus shown that $m_n \in U_n$, hence,
  $\init = m_n$ belongs to $\bigcup_k U_k$.

  \smallskip

  Let us show the converse of the lemma.
  We first prove, by induction on $k$,
  that $U_k \subseteq \prestarfinal$ for every $k \in \setN$.
  The basis follows from the observation that
  $U_0 \subseteq \upa{\final} \subseteq \prestarfinal$.
  For the induction step,
  let $k \in \setN$ and assume that $U_k \subseteq \prestarfinal$.
  Recall that
  $U_{k+1} = \upa{(\pre{U_k} \cap I)} \cup U_k$, hence,
  $U_{k+1} \subseteq \upa{\pre{U_k}} \cup U_k$.
  Since $\petri$ is a Petri net,
  $\pre{S}$ is upward-closed for every upward-closed subset $S \subseteq \pset$.
  It follows that
  $\upa{\pre{U_k}} = \pre{U_k}$.
  We derive from the induction hypothesis that
  $U_{k+1} \subseteq \pre{U_k} \cup U_{k} \subseteq \prestarfinal$.
  We have thus shown that $U_k \subseteq \prestarfinal$ for every $k \in \setN$.
  The observation that
  $\final \in \Cov \Leftrightarrow \init \in \prestarfinal$
  concludes the proof of the lemma.
  \qed
\end{proof}

\section{Proofs for \cref{sec:icover-algorithm}}
\label{appendix:sec:icover-algorithm}
\lemwqomin*
\begin{proof}
  The partial order $\leq$ on $\pset$ is a well-quasi-order by Dickson's Lemma.
  Therefore,
  the set $\Min S$ of minimal elements of $S$ is necessarily finite.
  Moreover,
  $S \subseteq \upa{\Min S}$ since $\leq$ is well-founded.
  It follows that $\upa{S} = \upa{\Min S}$.
  \qed
\end{proof}

\lemcpre*
\begin{proof}
  Let $u \in \pre[^t]{\upa{m}}$.
  There exists $v \geq m$ such that $u \transition{t} v$.
  Consider a place $p \in P$.
  It holds that $u(p) \geq F(p, t)$ and $v(p) = u(p) - F(p, t) + F(t, p)$
  since $u \transition{t} v$.
  We consider two cases.
  \begin{itemize}
  \item
    If $m(p) \leq F(p, t)$ then $\cpre[^t]{m}(p) = F(p, t) \leq u(p)$.
  \item
    If $m(p) \geq F(p, t)$ then $\cpre[^t]{m}(p) = F(p, t) + m(p) - F(t, p)$.
    Since $v \geq m$, we get that
    $\cpre[^t]{m}(p) \leq F(p, t) + v(p) - F(t, p) = u(p)$.
  \end{itemize}
  In both cases, we obtain that $\cpre[^t]{m}(p) \leq u(p)$.
  We have thus shown that $u \in \upa{\cpre[^t]{m}}$.

  \smallskip

  Conversely,
  let $u \in \upa{\cpre[^t]{m}}$.
  This means that $u(p) \geq \cpre[^t]{m}(p)$ for every place $p \in P$.
  Therefore,
  $u(p) \geq F(p, t)$ and $u(p) \geq F(p, t) + m(p) - F(t, p)$.
  It follows that $u \transition{t} v$ for the marking $v \geq m$
  defined by $v(p) = u(p) - F(p, t) + F(t, p)$.
  We have thus shown that $u \in \pre[^t]{\upa{m}}$.
  \qed
\end{proof}

\lemukbk*
\begin{proof}
  It is readily seen that $B_k$ and $P_k$ are finite subsets of $\pset$ for every $k$.
  We first observe that, for every $k$ with $0 \leq k < \ell_P$,
  $$
  \begin{array}{@{\hspace{18.5mm}}r@{\ \ }c@{\ \ }ll}
    \upa{P_k}
    & = & \upa{\left((\{\cpre[^t]{m} \mid t \in T, m \in B_k\} \setminus \upa{B_k}) \cap I\right)}
    & \hspace{7.5mm} [\text{Lines~\ref{line:assign:N}--\ref{line:assign:P}}]\\
    & = & \upa{\left(\{\cpre[^t]{m} \mid t \in T, m \in B_k\} \cap (I \setminus \upa{B_k})\right)}\\
    & = & \bigcup_{t \in T} \upa{\left(\cpre[^t]{B_k} \cap (I \setminus \upa{B_k})\right)}\\
    & = & \bigcup_{t \in T} \upa{\left(\pre[^t]{\upa{B_k}} \cap (I \setminus \upa{B_k})\right)}
    & \hspace{5mm} [\text{\cref{lem:cpre-inv}}]\\
    & = & \upa{\left(\pre{\upa{B_k}} \cap (I \setminus \upa{B_k})\right)}\\
    & = & \upa{\left((\pre{\upa{B_k}} \cap I) \setminus \upa{B_k}\right)}
  \end{array}
  $$
  Let us now prove, by induction on $k$, that $U_k = \upa{B_k}$
  for every $k$ with $0 \leq k < \ell_B$.
  The basis $U_0 = \upa{B_0}$ follows from
  lines $1$--\ref{line:while-start} of $\mathtt{ICover}$ and from the definition of $U_0$.
  For the induction step,
  let $k \in \setN$ with $k+1 < \ell_B$, and assume that $U_k = \upa{B_k}$.
  Line~\ref{line:while-end} entails that $B_{k+1} = \Min (B_k \cup P_k)$.
  It follows that
  $$
  \begin{array}{@{\hspace{24.5mm}}r@{\ \ }c@{\ \ }ll}
    \upa{B_{k+1}}
    & = & \upa{B_k} \cup \upa{P_k}
    & \hspace{11.5mm} [\text{\cref{lem:wqo-min}}]\\
    & = & \upa{B_k} \cup \upa{\left((\pre{\upa{B_k}} \cap I) \setminus \upa{B_k}\right)}\\
    & = & U_k \cup \upa{\left((\pre{U_k} \cap I) \setminus U_k\right)}
    & \hspace{13.5mm} [U_k = \upa{B_k}]\\
    & = & U_k \cup \upa{(\pre{U_k} \cap I)}\\
    & = & U_{k+1}
  \end{array}
  $$
  This concludes the proof that
  $U_k = \upa{B_k}$ for every $k$ with $0 \leq k < \ell_B$.
  Moreover,
  coming back to the characterization of $\upa{P_k}$,
  we get that
  $$
  \begin{array}{r@{\ \ }c@{\ \ }l}
    \upa{P_k}
    & = & \upa{\left((\pre{\upa{B_k}} \cap I) \setminus \upa{B_k}\right)}\\
    & = & \upa{\left((\pre{U_k} \cap I) \setminus U_k\right)}\\
    & = & \upa{(U_{k+1} \setminus U_k)}
  \end{array}
  $$
  for every $k$ with $0 \leq k < \ell_P$.
  \qed
\end{proof}

\thmicover*
\begin{proof}
  Let us first prove termination.
  We need to show that the unique maximal execution of $\mathtt{ICover}(\petri, \final, I)$
  is finite.
  By contradition, assume that $\ell_B = \infty$.
  According to \cref{lem:uk-stationary},
  there exists an index $h \in \setN$ such that $U_h = U_{h+1}$.
  We derive from \cref{lem:uk-bk} that $P_h = \emptyset$.
  Therefore,
  the execution should terminate at line~\ref{line:return:false} during
  the $(h+1)^\mathrm{th}$ iteration of the $\mathbf{while}$ loop.
  This contradicts our assumption that $\ell_B = \infty$.

  \smallskip

  We now turn our attention to the correctness of $\mathtt{ICover}$.
  As it is finite,
  the unique maximal execution of $\mathtt{ICover}(\petri, \final, I)$
  either returns $\mathtt{False}$ at line~\ref{line:return:false} or
  returns $\mathtt{True}$ at line~\ref{line:return:true}.
  \begin{itemize}
  \item
    If it returns $\mathtt{False}$ then $P_{\ell_P - 1} = \emptyset$
    and it follows from \cref{lem:uk-bk} that
    $U_{\ell_P} \subseteq U_{\ell_P-1}$.
    We get from the definition of $\Uk$ that
    $U_k = U_{\ell_P-1}$ for every $k \geq \ell_P$.
    Therefore, $U_{\ell_P-1} = \bigcup_k U_k$.
    Moreover,
    $\init \not\in \upa{B_{\ell_P-1}}$ because the condition of the
    $\mathbf{while}$ loop had to hold.
    It follows from \cref{lem:uk-bk} that $\init \not\in U_{\ell_P-1}$.
    We derive from \cref{lem:uk-cov} that $\final \not\in \Cov$.
  \item
    If it returns $\mathtt{True}$ then $\init \in \upa{B_{\ell_B - 1}}$
    and it follows from \cref{lem:uk-bk} that $\init \in U_{\ell_B - 1}$.
    We derive from \cref{lem:uk-cov} that $\final \in \Cov$.
  \end{itemize}
  \qed
\end{proof}

\section{Proofs for \cref{sec:sign-analysis}}
\label{appendix:sec:sign-analysis}
\lemprop*
\begin{proof}
  We can assume without loss of generality that $\bigwedge_{p\in P\moins Q}F(p,t)=0$ since otherwise the set $\prop{t}{Q}$ is empty. For the same reason, we can assume that there exists $q\in P$ such that $F(t,q)>0$. Let us prove that such a place $q$ cannot be in $Z$. We introduce the markings $m_t$ and $m_t'$ defined by $m_t(p)=F(p,t)$ and $m'_t(p)=F(t,p)$ for every $p\in P$. Those markings are the minimal ones satisfying $m_t\xrightarrow{t}m_t'$. Observe that for every $p\in Z$, we have $p\in P\moins Q$ since $Q\cap Z$ is empty. It follows that $F(p,t)=0$ for every $p\in Z$. Hence $m_t$ is in the inductive invariant $I_Z$. Since $m_t\xrightarrow{t}m_t'$, we deduce that $m_t'\in I_Z$. As $F(t,q)>0$, we get $m'_t(q)>0$. Hence $q\not\in Z$. 
  \qed
\end{proof}

\lemzpmoinsq*
\begin{proof}
  Since $Q_0\subseteq P\moins Z$, \cref{lem:prop} shows by induction that $Q_k\subseteq P\moins Z$ for every $k$. It follows that $Q\subseteq P\moins Z$. We derive $Z\subseteq P\moins Q$. The converse inclusion is obtained by proving that  the set $M=\{m\in\setN^P \mid \bigwedge_{p\in P\moins Q}m(p)=0\}$ is an inductive invariant. First of all, since $Q_0\subseteq Q$, we deduce that $\init\in M$. Now let us consider $m\in M$ and a transition $t$ such that $m\xrightarrow{t}m'$ for some marking $m'$. Observe that $m(p)\geq F(p,t)$ for every $p\in P$. In particular, for $p\in P\moins Q$, the equality $m(p)=0$ implies $F(p,t)=0$. Assume by contradiction that $m'\not\in M$. In that case, there exists $q\in P\moins Q$ such that $m'(q)>0$. Since $m'(q)=m(q)+F(t,q)+F(q,t)$ and $m(q)=0=F(q,t)$, we deduce that $F(t,q)>0$. Thus $q\in\prop{t}{Q}$. By definition of $Q$, we get $q\in Q$ and we obtain a contradiction. Thus $M$ is an inductive invariant. By maximality of $Z$, we get the inclusion $P\moins Q\subseteq Z$. Thus $Z=P\moins Q$.
  \qed
\end{proof}

\end{document}